\documentclass[10pt,psamsfonts]{amsart}
\usepackage{amsmath}
\usepackage{amsthm}
\usepackage{amssymb}
\usepackage{amscd}
\usepackage{amsfonts}
\usepackage{amsbsy}
\usepackage{graphicx}
\usepackage[dvips]{psfrag}
\usepackage{array}
\usepackage{hyperref}

\newcommand{\R}{\ensuremath{\mathbb{R}}}
\newcommand{\N}{\ensuremath{\mathbb{N}}}

\newcommand{\C}{\ensuremath{\mathbb{C}}}

\newcommand{\cpe}{\mathrm{EPC}}
\newcommand{\cce}{\mathrm{ECC}}
\newcommand{\cme}{\mathrm{EMC}}

\newcommand{\defi}{\overset{\underset{\mathrm{def}}{}}{=}}

\newcommand\captionof[1]{\def\@captype{#1}\caption}

\newcolumntype{L}{>{\displaystyle}l}
\newcolumntype{C}{>{\displaystyle}c}
\newcolumntype{R}{>{\displaystyle}r}
\renewcommand{\arraystretch}{1.5}

\newtheorem {theorem} {Theorem} 

\newtheorem {lemma} [theorem] {Lemma}

\begin{document}

\author[I. G. Cesca and D.D. Novaes]
{Igor G. Cesca$^1$ and Douglas D. Novaes$^2$}

\address{$^1$ Departamento de Engenharia do Petr\'{o}leo, Universidade
Estadual de Campinas, Caixa Postal 6052, 13083--970, Campinas, SP,
Brazil} \email{igcesca@gmail.com}

\address{$^2$ Departamento de Matematica, Universidade
Estadual de Campinas, Caixa Postal 6065, 13083--859, Campinas, SP,
Brazil} \email{ddnovaes@gmail.com}

\title[Physical assets replacement: an analytical approach]
{Physical assets replacement: an analytical approach}



\maketitle

\begin{abstract}
The economic life of an asset is the optimum length of its usefulness, which is the moment that the asset's expenses are minimum. In this paper, the economic life of physical assets, such as industry machine and equipment, can be interpreted as the moment that the minimum is reached by its equivalent property cost function, defined as the sum of all equivalent capital and maintenance costs during its life.

Many authors in classical papers have used principles of engineering economic to solve the assets replacement problem. However, in the literature, the main attributes found were proved with intuitive ideas instead mathematical analysis. Therefore, in this paper the main goal is to study these principles of engineering economic with mathematical techniques. 

Here, is used non-smooth analysis to classify all the possibilities for the minimum of a class of equivalent property cost functions of assets. The minimum of these function gives the optimum moment for the asset to be replaced, i.e., its economic life. 
\end{abstract}

\section{Introduction to physical assets replacement}\label{s1}

Physical assets, such as industry equipment, are vulnerable to devaluation and obsolescence. Among the many consequences of devaluation, there are, for instance, output decline, operation and maintenance expenses increase. These causes can lead to market value declination (see Park and Sharpe-Bette in \cite{PS}). Regarding asset's obsolescence, the main causes are technological innovation and change inside the company organization. Therefore, the asset is no longer needed (see Park in \cite{P}).

\smallskip

Nevertheless, with the proper maintenance, physical assets can be used for much more time than its physical nature allow. For example, it is possible to see vintage cars driving on the streets. However, to make it possible, the companies must be willing to pay a higher price.

\smallskip

Along with that comes the concept of {\it Economic Life}. An asset's economic life is the length of its usefulness, in a way that the expenses, i.e., annually sum of the maintenance costs and capital costs, are minimum. Therefore, the economic life of the asset is the optimum moment to replace the asset.

\smallskip

Because of that, if the asset is kept longer than its economic life, the expenses of maintenance will have increase a lot. Meanwhile, if the asset is replaced before its economic life, the capital cost will not have been fully fiscal depreciated. Therefore, part of the investment, in the acquisition cost of the asset will be lost. So, physical assets in general are always used for a limited time.

\section{Objectives}\label{s2}

Many authors, in classical papers, like Alchian in \cite{A}; Park in \cite{P}; Park and Sharp-Bette in \cite{PS}; Grant, et al in \cite{GIL}; and Thuesen, et al in \cite{TFT}, have used principles of engineering economic, as net present value and annuity equivalent, in the assets replacement problem.

\smallskip

However, in the literature, the main attributes found were proved with intuitive ideas and with no mathematical accuracy. Therefore, in this paper the main goal is to study these principles of engineering economic with mathematical analysis.

\section{Mathematical modeling}\label{s3}

The concept of economic life can be modeled with the following equations:

\subsection[Equivalent Capital Cost]{Equivalent Capital Cost\protect\footnote{The {\it Capital Cost Acquisition} is the cash outflow regarding asset's purchase. This cost also include shipping cost, installation cost and training cost. Besides, this is also called first cost or investment cost, because this kind of expense is the cost of getting an activity or project started (see Fabrycky and Blanchard in \cite{FB}, p. 22).}}

Two components are involved in the evaluation of capital cost: {\it Acquisition Cost} $A>0$, and {\it Salvage Value} $R(t)$. The first is a fixed value, while the second changes along the time. The result obtained from the difference between the annual equivalent and each component is the equivalent capital cost (see Appendix A), as showed in equation \eqref{cce}.

\begin{equation}\label{cce}
\cce\defi g(t)=\dfrac{e^r-1}{e^{r t}-1} \left(A e^{r t}-R(t)\right).
\end{equation}

Here, $0<r\leq1$ is the nominal interest rate. It is usual to see the value around $10\%$.

\subsection[Equivalent Maintenance Cost]{Equivalent Maintenance Cost\protect\footnote{On this category the main costs are due to cleaning. lubrication, components adjustment and repair, among others (see Hastings in \cite{H}).}}

To compute the equivalent maintenance cost, firstly it is necessary to evaluate the sum of present value of the series of {\it Maintenance Costs} $M(t)$ along the time. Then compute the present value in annually equivalents, as it is showed in equation \eqref{cme}.

\begin{equation}\label{cme}
\cme\defi f(t)=\dfrac{e^r-1}{e^{r t}-1} e^{r t} \int_0^t M(s) e^{-r s} \, ds.
\end{equation}

\subsection{Equivalent Property Cost}

The equivalent property cost of an asset is the sum of all capital and maintenance equivalent costs during its economic life, as it can be seen in equation \eqref{cpe}.

\begin{equation}\label{cpe}
\cpe\defi h(t)=f(t)+g(t).
\end{equation}

It is possible to see in Figure \ref{curve} the behavior of the functions \eqref{cce}, \eqref{cme}, and \eqref{cpe} along the time.

\smallskip

\begin{figure}[h]
\includegraphics[height=6.5cm]{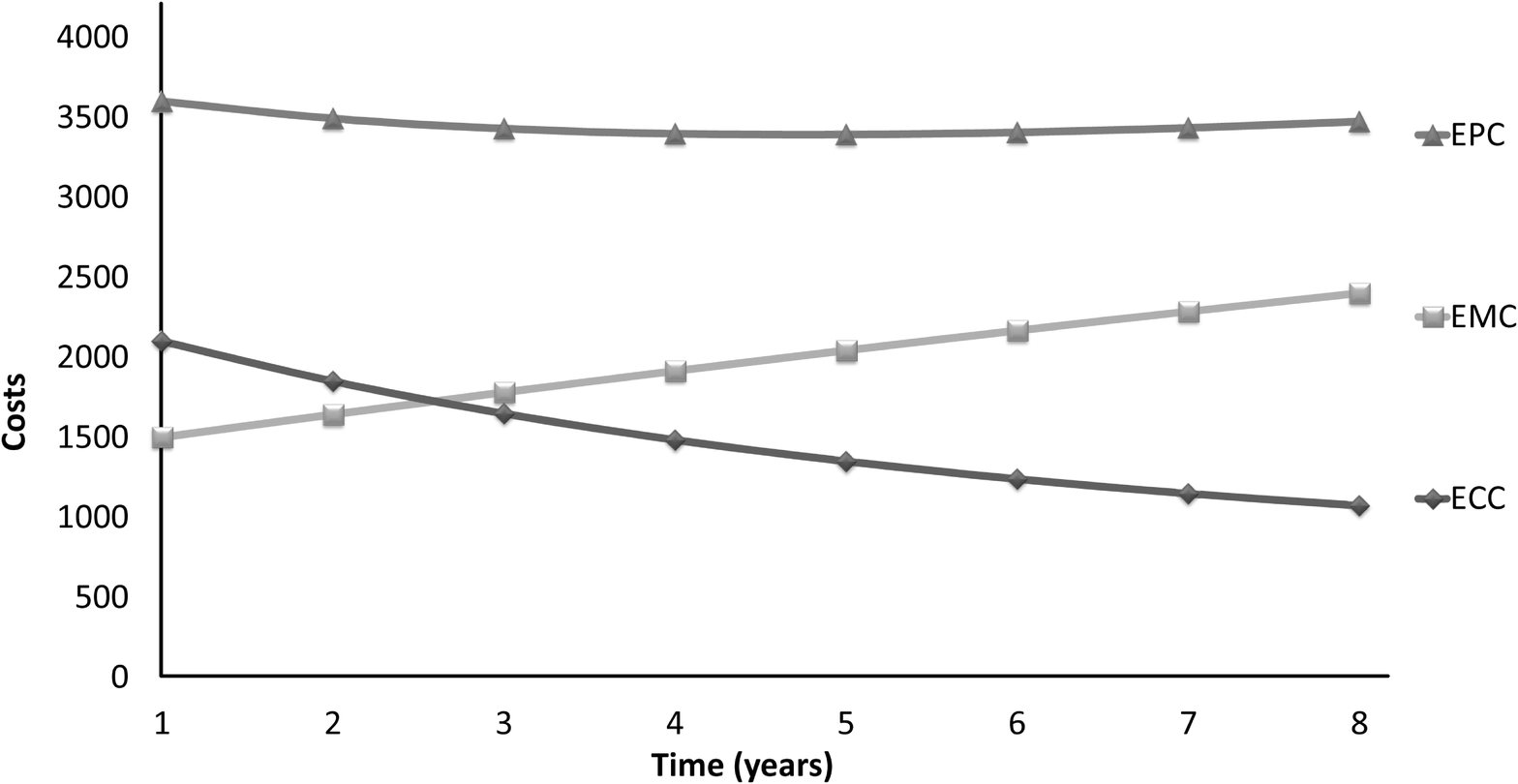}
\vskip 0cm \centerline{} \caption{\small \label{curve} Model to estimate the economic life.}
\end{figure}

\smallskip

The Figure \ref{curve} is representing the costs ($y$--axis) of the functions, \eqref{cce}, \eqref{cme} and \eqref{cpe}, along the time ($x$--axis). It can be seen that the equivalent maintenance costs increase, while the equivalent capital cost decrease. Therefore, the sum of both variables, i.e., the equivalent property cost, reaches a minimum point, which is the asset's economic life. In Figure \ref{curve}, the minimum point is reached at year five. So, it is the optimum moment for the asset be kept.

\smallskip

On balance, the consequences of a late replacement imply in high expenditures of operational and maintenance costs. Besides, there is also the opportunity cost of loss of asset's market value. On the other hand, an early replacement imply in selling the asset before its capital recovery .

\section{Physical Asset Model}\label{s4}

The maintenance costs function $M(t)$ is usually taken as a linear regression from a data collect regarding maintenance costs within the company.
Therefore, in many situations, it can be accepted that $M(t)$ is linear increasing function (see \eqref{M}), since the maintenance cost are higher each year.

\begin{equation}\label{M}
M(t)=at,
\end{equation}
with $a>0$.

\smallskip

The salvage value function $R(t)$ is the market value of the asset. It can be expected that $R(t)$ is decreasing, since physical assets are vulnerable to devaluation. However, the market value is always positive. So, given $A$, the acquisition cost of the asset, it is possible to assume, in many situations, that the asset value is depreciated every year by a tax $b$, i.e., $R(t)= A-bt$, until its fully depreciate at $t=A/b$. Thereafter, the asset value is constant equal zero (see \eqref{R}).

\begin{equation}\label{R}
R(t)=\left\{
\begin{array}{LCC}
A-bt&\textrm{if}& 0<t<A/b,\\
0&\textrm{if}& t\geq A/b,\\
\end{array}\right.
\end{equation}
with $A>0$ and $b>0$.

\smallskip

With the assumptions above, the equivalent property cost function \eqref{cpe} becomes
\begin{equation}\label{cpe1}
\renewcommand{\arraystretch}{2}
h(t)=\left\{
\begin{array}{LCC}
\dfrac{\left(e^r-1\right)}{r^2}\left(\frac{r t (b r-a)}{e^{r t}-1}+a+A r^2\right)&\textrm{if}& 0<t<A/b,\\
\dfrac{\left(-1+e^r\right)}{\left(-1+e^{r t}\right) r^2} \left(e^{r t} \left(a+A r^2\right)-a (1+r t)\right)&\textrm{if}& t\geq A/b,\\
\end{array}\right.
\end{equation}
and its derivative, for $t\neq A/b$, is given by
\begin{equation}\label{Dcpe1}
\renewcommand{\arraystretch}{2}
h'(t)=
\left\{
\begin{array}{LCC}
\dfrac{\left(e^r-1\right)}{r \left(e^{r t}-1\right)^2} \left(e^{r t} (r t-1)+1\right) (a-b r)&\textrm{if}& 0<t<A/b,\\
\dfrac{\left(e^r-1\right)}{r \left(e^{r t}-1\right)^2} \left(a-e^{r t} \left(-a r t+a+A r^2\right)\right)&\textrm{if}& t>A/b.
\end{array}\right.
\end{equation}

\section{Statements of the main results}

For $u>0$, define the function
\[
T(u)=1+u+W_0\left(-e^{-1-u}\right),
\]
where $W_0$ represents the main branch of the multivalued {\it Lambert W--Function} (see Appendix B).

\smallskip

Our main results, that classify all the possibilities for the minimum of the function $h(t)$, are the following:

\begin{theorem}\label{t1} If $c=Ar^2/a$, then we have the following possibilities for the minimum of the function $h(t)$:
\begin{itemize}
\item [($C_1$)] if
\[
\begin{array}{CCC}
a>b\,r &\textrm{and}& a\geq\dfrac{Abr^2}{Ar+b\left(e^{-\frac{Ar}{b}}-1\right)},
\end{array}
\]
then the minimum of the function $h(t)$, for $t\geq0$, is uniquely reached at $t=0$.

\smallskip

\item [($C_2$)] if
\[
\begin{array}{CCC}
a=b\,r &\textrm{and}& a\geq\dfrac{Abr^2}{Ar+b\left(e^{-\frac{Ar}{b}}-1\right)},
\end{array}
\]
then the minimum of the function $h(t)$, for $t\geq0$, is reached at all $t\in[0,A/b]$.

\smallskip

\item [($C_3$)] if
\[
\begin{array}{CCC}
a<b\,r &\textrm{and}& a\geq\dfrac{Abr^2}{Ar+b\left(e^{-\frac{Ar}{b}}-1\right)},
\end{array}
\]
then the minimum of the function $h(t)$, for $t\geq0$, is uniquely reached at $t=A/b$.

\smallskip

\item [($C_4$)] if
\[
\begin{array}{CCC}
a>b\,r &\textrm{and}& a<\dfrac{Abr^2}{Ar+b\left(e^{-\frac{Ar}{b}}-1\right)},
\end{array}
\]
then the minimum of the function $h(t)$, for $t\geq0$, is reached at $t=0$, or $t=T(c)/r$.

\smallskip

\item [($C_5$)] if
\[
\begin{array}{CCC}
a\leq b\,r &\textrm{and}& a<\dfrac{Abr^2}{Ar+b\left(e^{-\frac{Ar}{b}}-1\right)},
\end{array}
\]
then the minimum of the function $h(t)$, for $t\geq0$, is uniquely reached at $t=T(c)/r$.
\end{itemize}
\end{theorem}

\smallskip

The proof of Theorem \ref{t1} is given in Section 7. The next theorem divides the case $C_4$, of Theorem \ref{t1}, in three different cases, $C_4^1$, $C_4^2$ and $C_4^3$.  

\smallskip

\begin{theorem}\label{t2}
Considering the assumptions of the case $C_4$ of Theorem \ref{t1}, follows:
\begin{itemize}
\item[($C_4^1$)] If
\[
A>\log\left(1-\dfrac{br}{a}\right)^{-\frac{a}{r^2}}-\dfrac{b}{r},
\]
then the minimum of the function $h(t)$, for $t\geq0$, is uniquely reached at $t=0$;

\smallskip

\item[($C_4^2$)] if
\[
A=\log\left(1-\dfrac{br}{a}\right)^{-\frac{a}{r^2}}-\dfrac{b}{r},
\]
then the minimum of the function $h(t)$, for $t\geq0$, is reached at $t=0$ and $t=T(c)/r$;

\smallskip

\item[($C_4^3$)] if
\[
A<\log\left(1-\dfrac{br}{a}\right)^{-\frac{a}{r^2}}-\dfrac{b}{r},
\]
then the minimum of the function $h(t)$, for $t\geq0$, is uniquely reached at $t=T(c)/r$.
\end{itemize}
\end{theorem}

\smallskip

The proof of Theorem \ref{t2} is given in Section 7. 

\smallskip

Resuming the results of Theorems \ref{t1} and \ref{t2}, it follows the classification:

\bigskip

\begin{itemize}
\item in the cases $C_1$ and $C_4^1$ (see Figure \ref{A1_41}), the minimum is uniquely reached at $t=0$;

\smallskip

\begin{figure}[h]
\psfrag{h0}{$h_0$}
\psfrag{h}{$h(t)$}
\psfrag{t}{$t$}
\psfrag{T}{$T(c)/r$}
\psfrag{A}{$A/b$}
\includegraphics[height=4cm]{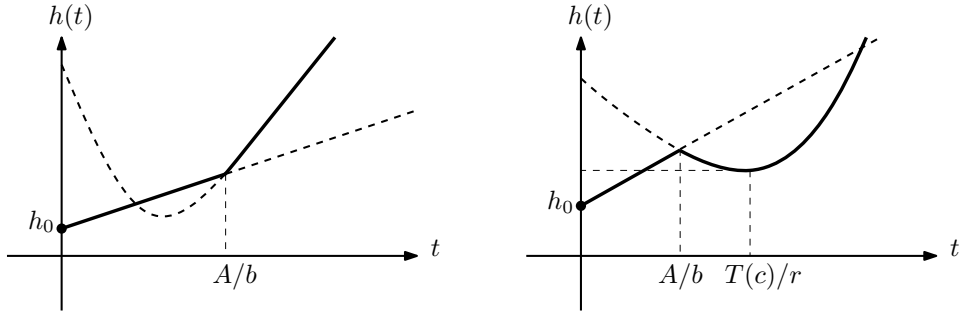}
\vskip 0cm \centerline{} \caption{\small \label{A1_41} Cases $C_1$ and $C_4^1$.}
\end{figure}

\smallskip

\item in the case $C_2$ (see Figure \ref{A2}), the minimum is reached at all $t\in[0\,,\,A/b]$;

\smallskip

\begin{figure}[h]
\psfrag{h0}{$h_0$}
\psfrag{h}{$h(t)$}
\psfrag{t}{$t$}
\psfrag{T}{$T(c)/r$}
\psfrag{A}{$A/b$}
\includegraphics[height=4cm]{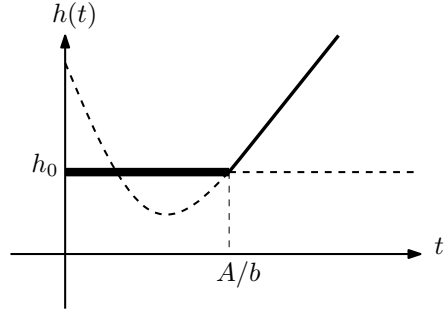}
\vskip 0cm \centerline{} \caption{\small \label{A2} Case $C_2$.}
\end{figure}

\smallskip

\item in the case $C_3$ (see Figure \ref{A3}), the minimum is uniquely reached at $t=A/b$;

\smallskip

\begin{figure}[h]
\psfrag{h0}{$h_0$}
\psfrag{h}{$h(t)$}
\psfrag{t}{$t$}
\psfrag{A}{$A/b$}
\includegraphics[height=4cm]{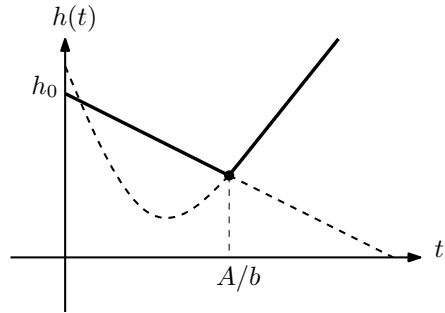}
\vskip 0cm \centerline{} \caption{\small \label{A3} Case $C_3$.}
\label{A3}
\end{figure}

\smallskip

\item in the case $C_4^2$ (see Figure \ref{A42}), the minimum is reached at two points $t=0$ and $t=T(c)/r$;

\smallskip

\begin{figure}[h]
\psfrag{h0}{$h_0$}
\psfrag{h}{$h(t)$}
\psfrag{t}{$t$}
\psfrag{T}{$T(c)/r$}
\psfrag{A}{$A/b$}
\includegraphics[height=4cm]{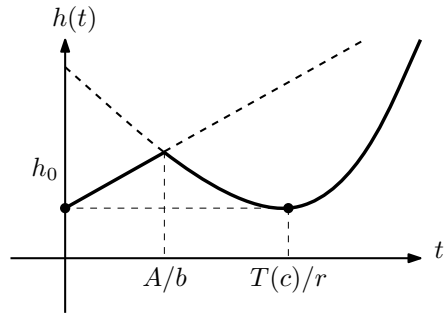}
\vskip 0cm \centerline{} \caption{\small \label{A42} Case $C_4^2$.}
\end{figure}

\smallskip

\item in the cases $C_4^3$ and $C_5$ (see Figure \ref{A43_5}), the minimum is uniquely reached at $t=T(c)/r$;

\smallskip

\begin{figure}[h]
\psfrag{h0}{$h_0$}
\psfrag{h}{$h(t)$}
\psfrag{t}{$t$}
\psfrag{T}{$T(c)/r$}
\psfrag{A}{$A/b$}
\includegraphics[height=4cm]{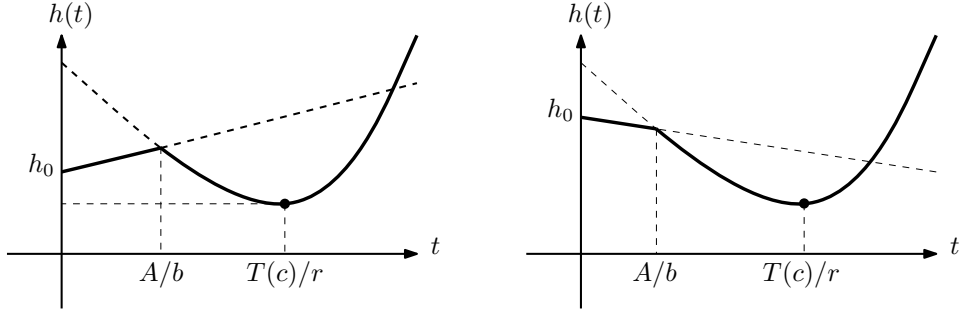}
\vskip 0cm \centerline{} \caption{\small \label{A43_5} Cases $C_4^3$ and $C_5$.}
\end{figure}

\end{itemize}


\section{Basic lemmas}

Before proving the Theorems \ref{t1} and \ref{t2}, we shall prove two lemmas about the critical points of the function \eqref{cpe1} on the intervals $I_1=(0,A/b)$ and $I_2=(A/b,+\infty)$. For this aim, we study the zeros of the derivative \eqref{Dcpe1}. Observe that the functions $h(t)$ and $h'(t)$, given respectively in \eqref{cpe1} and \eqref{Dcpe1}, restricted to the intervals $I_1$ and $I_2$, are smooth.

\begin{lemma}\label{l1}
If $I_1=(0,A/b)$, then $h'(t)\neq0$ for $t\in I_1$. Moreover, if $a>br$, then $h(t)$ is increasing in $I_1$; if $a<br$, then $h(t)$ is decreasing in $I_1$; and if $a=br$, then $h(t)$ is constant in $I_1$.
\end{lemma}
\begin{proof}


For $t\in I_1$
\[
h'(t)=\dfrac{\left(e^r-1\right)}{r \left(e^{r t}-1\right)^2}p_1(rt)
\]
where $p_1(\tau)=(a-br) \left(1+e^{\tau} (\tau-1)\right)$. Thus for $t\in I_1$, we have that $h'(t)=0$ if and only if $p_1(rt)=0$.

On the other hand, for $a\neq br$, $p_1(\tau)=0$ if and only if $\tau=0$. Therefore, for $t\in I_1$ and $br\neq a$, $h'(t)=0$ if and only if $t=0$.

Now, observe that, for $t\in I_1$,
\[
h'(t)\left\{\begin{array}{CCC}
>0, &\textrm{if}& a>br,\\
=0, &\textrm{if}& a=br,\\
<0, &\textrm{if}& a<br.
\end{array}\right.
\]

Moreover, if $a=br$, then, for $t\in I_1$,
\[
h(t)=\dfrac{e^r-1}{r}\left(b+Ar\right)
\]
for $0<t<A/b$, which concludes the proof.
\end{proof}

\begin{lemma}\label{l2} For $t> A/b$ follows:
\begin{itemize}
\item[(i)]if
\[
a\geq\dfrac{Abr^2}{Ar+b\left(e^{-\frac{Ar}{b}}-1\right)},
\]
then $h'(t)\neq0$, for $t> A/b$; Moreover, $h(t)$ is increasing for $t> A/b$;

\smallskip

\item[(ii)]if
\[
a<\dfrac{Abr^2}{A+b\left(e^{-\frac{A}{b}}-1\right)},
\]
then $h'(t)=0$ for $t=T(c)/r$, where
\[
1+c+W_0\left(-e^{-1-c}\right),
\]
where $c=Ar^2/a$. Moreover, $h(t)$ is decreasing for $A/b<t<T(c)$; and $h(t)$ is increasing for $t>T(c)$;
\end{itemize}
\end{lemma}

\begin{proof}
For $t >A/b$, 
\[
h'(t)=\dfrac{\left(e^r-1\right)}{r \left(e^{r t}-1\right)^2}p_2(rt)
\]
where $p_2(\tau)=a-e^{\tau} (a+Ar^2-a\tau)$. Thus for $t >A/b$, we have that $h'(t)=0$ if and only if $p_2(r\,t)=0$.

On the other hand, $p_2(r\,t)=0$ if and only if $p(r\,t)=c$, where
\[
\begin{array}{CCC}
p(\tau)=\dfrac{\tau e^\tau-e^{\tau}+1}{e^{\tau}} & \textrm{and} & c=\dfrac{Ar^2}{a}.
\end{array}
\]

Note that the function $p(\tau)$ is increasing for $\tau>0$. Indeed $p'(\tau)=1-\cosh(\tau)+\sinh(\tau)$. Thus if $p(r\,t_0)=c$, then $t_0>A/b$ if and only if
\[
\begin{array}{CCCCCC}
p(r\,t_0)>p\left(r\,\dfrac{A}{b}\right)&\Leftrightarrow& \dfrac{Ar^2}{a}>\dfrac{A r}{b}+e^{-\frac{A r}{b}}-1 &\Leftrightarrow & a<\dfrac{Abr^2}{A+b\left(e^{-\frac{A}{b}}-1\right)}.
\end{array}
\]

Hence if 
\[
a\geq\dfrac{Abr^2}{Ar+b\left(e^{-\frac{Ar}{b}}-1\right)},
\]
then $h'(t)\neq 0$ for $t>A/b$.

Now, assume that
\[
a<\dfrac{Abr^2}{Ar+b\left(e^{-\frac{Ar}{b}}-1\right)}.
\]

Using the software {\it Mathematica 8.0}, we obtain that: if
\[
T(u)=1+u+W_0\left(-e^{-1-u}\right),
\]
where $W_0$ represents the main branch of the multivalued {\it Lambert W--Function} (see Appendix B), then
\[
p\left(r\dfrac{T(c)}{r}\right)=c.
\]

Again, using the software {\it Mathematica 8.0}, we can check that $T(u)>0$ for $u>0$, which implies, as we have seen, that $T(c)/r>A/b$. 

Now, using the fact that the function $p(\tau)$ is increasing for $\tau>0$, we have, in this case, that $t=T(c)/r$ is the unique critical point of the function $h(t)$ for $t>A/b$. Moreover, is easy to see that $p(r\,t)<c$ for $t<T(c)/r$, and $p(r\,t)>c$ for $t>T(c)/r$. Thus the derivative $h'(t)$ is negative for $A/b<t<T(c)/r$, and is positive for $t>T(c)/r$. Hence we can conclude that $h(t)$ is decreasing for $A/b<t<T(c)/r$; and $h(t)$ is increasing for $t>T(c)/r$; which concludes the proof.
\end{proof}

\section{Proofs of Theorems \ref{t1} and \ref{t2}}

Now, we are ready to prove the Theorems \ref{t1} and \ref{t2}

\begin{proof}[Proof of Theorem \ref{t1}]

The function \eqref{cpe1} is continuous for $t\geq 0$ and piecewise differentiable. Thus we can find the minimum of \eqref{cpe1} comparing the minimums of each piecewise. In other words, we have to choose the minimum among the values of the function evaluated at the critical points of each piece and at the extreme points $t=0$ and $t=A/b$.

\smallskip

In the Case $C_1$, by Lemmas \ref{l1} and \ref{l2}, the function $h(t)$ is increasing for $t\geq 0$. Hence the minimum is uniquely reached at $t=0$.

\smallskip

In the Case $C_2$, by Lemma \ref{l1}, the function $h(t)$ is constant for $0\leq t<\leq A/b$; and by Lemma \ref{l2}, the function $h(t)$ is increasing for $t\geq A/b$. Hence the minimum is reached at all $t\in\left[0,A/b\right]$.

\smallskip

In the Case $C_3$, by Lemma \ref{l1}, the function $h(t)$ is decreasing for $0\leq t<\leq A/b$; and by Lemma \ref{l2}, the function $h(t)$ is increasing for $t\leq A/b$. Hence the minimum is uniquely reached at $t=A/b$.

\smallskip

In the Case $C_4$, by Lemma \ref{l1}, the function $h(t)$ is increasing for $0\leq t\leq A/b$, thus the minimum for $0<t<A/b$ is reached at $t=0$. On the other hand, by Lemma \ref{l2} there is a local minimum in $t=T(c)/r>A/b$. Moreover, by Lemma \ref{l2}, the function $h(t)$ is increasing for $t\geq T(c)/r$. Hence the minimum is reached at $t=0$, or $t=T(c)/r$.

\smallskip

In the Case $C_5$, by Lemma \ref{l1}, the function $h(t)$ is decreasing for $0\leq t\leq A/b$, and by Lemma \ref{l2}, there is a local minimum in $t=T(c)/r>A/b$. Moreover, by Lemma \ref{l2}, the function $h(t)$ is increasing for $t\geq T(c)/r$. Hence the minimum is uniquely reached at $t=T(c)/r$.

\end{proof}

\begin{proof}[Proof of Theorem \ref{t2}]

Consider the assumptions of the case $C_4^1$. It is clear that the minimum of the function $h(t)$, for $t\geq0$, is reached at $t=0$ if and only if $h(0)<h(T(c)/r)$. We know that
\[
\begin{array}{CCC}
h(0)=\dfrac{e^r-1}{r}\left(Ar+b\right)& \textrm{and} & h\left(\dfrac{T(c)}{r}\right)=\dfrac{e^r-1}{r^2}\left(a+Ar^2+aW_0\left(-e^{-1-\frac{Ar^2}{a}}\right)\right).
\end{array}
\]

Thus $h(0)<h(T(c)/r)$ if and only if
\begin{equation}\label{i1}
-1<\dfrac{br}{a}-1<W_0\left(-e^{-1-\frac{Ar^2}{a}}\right).
\end{equation}

The function $W_0:(-1/e,\infty)\longrightarrow(-1,\infty)$ is invertible and $W_0^{-1}:(-1,\infty)\longrightarrow(-1/e,\infty)$ is an increasing function (see Appendix B). Then we can apply the function $W_0^{-1}$ in both sides of the inequality \eqref{i1}, since $br/a-1>-1$, to obtain the equivalent inequality

\begin{equation}\label{i2}
\left(1-\dfrac{br}{a}\right)e^{\dfrac{br}{a}-1}>e^{-1-\frac{Ar^2}{a}}>0.
\end{equation}

Now, applying the $\log$ function in both sides of the inequality \eqref{i2}, we can conclude that $h(0)<h(T(c)/r)$ if and only if
\[
A>\log\left(1-\dfrac{br}{a}\right)^{-\frac{a}{r^2}}-\dfrac{b}{r}.
\]

Hence we conclude the proof of theorem for the case $C_4^1$. The proofs for the cases $C_4^2$ and $C_4^3$, are completely analogous.

\end{proof}

\section{Conclusion and future directions}

The methodology defined in the classical literature, assets replacement problem involve solving it organizing and analyzing tables and equations, as it can be verified. However, this involve too much work and effort. 
So, for assets with the characteristics described in Section \ref{s4} it can be applied the methodology described in this paper. With the analytical approach, it takes less time and effort to see the when is the economic life and to replace the asset. Following this article approach, through Theorem \ref{t1} and Theorem \ref{t2} results, it can simply solve the assets replacement problem.

\medskip

The function $M(t)$ in the expression \eqref{cme} comes from a linear regression. Thus in a more general case, $M(t)$ could be a piecewise function composed by affine and exponential translated functions. For instance (see Figure \ref{fM}), we can take
\begin{equation}\label{M1}
M(t)=\sum_{i\in N}\chi_{I_i}(t)M_i(t),
\end{equation}
where, $N$ is a collection of natural indexes; $(I_i)_{i\in\N}$ is a sequence of disjoint real intervals such that covers the real line $\R$; for each $i\in N$, $M_i$ is an affine, or exponential translated function; and $\chi_I:\R\rightarrow \{0,1\}$ is the characteristic function, i.e., if $I$ is a subset of $\R$, then
\[
\chi_I(t)=\left\{\begin{array}{LCL}
1,&\textrm{if}&t\in I,\\
0,&\textrm{if}&t\notin I.
\end{array}\right.
\]

We intend, in future works, deal with the problem of assets replacement considering maintenance costs functions as given in \eqref{M1}.

\begin{figure}[h]
\psfrag{M}{$M(t)$}
\psfrag{M1}{$M_1$}
\psfrag{M2}{$M_2$}
\psfrag{t}{$t$}
\psfrag{I1}{$I_1$}
\psfrag{I2}{$I_2$}
\psfrag{P}{$\cdots$}
\includegraphics[height=5cm]{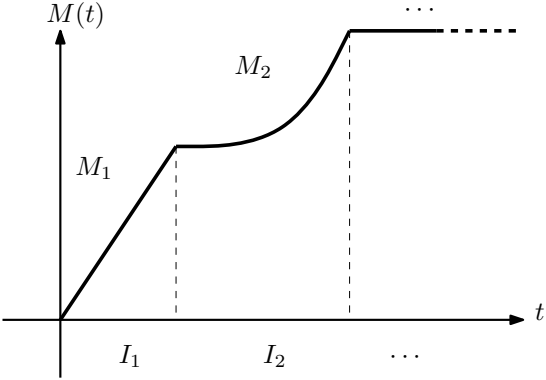}
\vskip 0cm \centerline{} \caption{\small \label{fM} Piecewise function $M(t)$.}
\end{figure}

\section*{Appendix A: Financial Criterion Equivalence}\label{A}
Given a series of irregular cash flow from $t=1$ to $N$ , we want to make it a regular cash flow along the years. The scenario is possible to see in Figure \ref{Cash} below.

\smallskip

\begin{figure}[h]
\psfrag{1}{1}
\psfrag{2}{2}
\psfrag{3}{3}
\psfrag{4}{4}
\psfrag{N-1}{$N-1$}
\psfrag{N}{$N$}
\includegraphics[width=8cm]{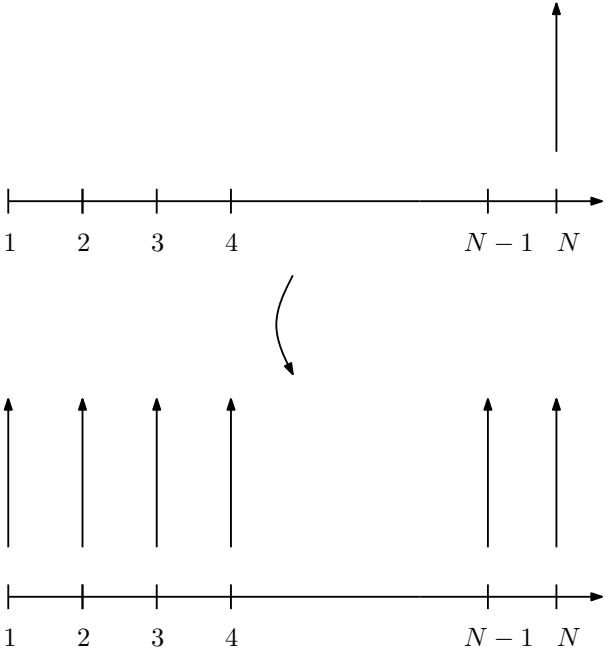}
\vskip 0cm \centerline{} \caption{\small \label{Cash} Equivalence of cash flows.}
\end{figure}

\smallskip

In Figure \ref{Cash}, it is possible to see that first we have a cash flow with different values along the years. Then, we have a new cash flow with equal deposits along the years. Now, the question is how to make this equivalence?

\smallskip

Let us consider a future deposit $F$ at time $N$  and a series of regular deposit $A$ from year $1$  to $N$. Then, to make the equivalence, it is necessary to consider the interest $i$  for each year. Therefore,
\begin{equation}\label{F}
F=A\sum_{k=1}^{N}(1+i)^{N-k}.
\end{equation}

Multiplying equation \eqref{F} by $(1+i)$, follows

\begin{equation}\label{F2}
(1+i)F=A\sum_{k=1}^{N}(1+i)^{1+N-k}.
\end{equation}

Now, computing the difference between equations \eqref{F2} and \eqref{F}, we have that 
\[
iF=A((1+i)^N-1).
\]

Since $F=P(1+i)^N$, we can conclude that
\begin{equation}\label{P}
\begin{array}{CCC}
P=\dfrac{A((1+i)^N-1)}{i(1+i)^N} & \Leftrightarrow & A=\dfrac{P i(1+i)^N}{(1+i)^N-1}
\end{array}
\end{equation}

As a consequence of equivalence \eqref{P}, we can take any present value and transforming in a series of regular cash flows. 

\smallskip

Consider $r$ the nominal interest, $M$ the number of period per year, and $i$ the real interest. It is possible to establish the equivalence

\begin{equation}\label{I}
i=\left(1+\dfrac{r}{M}\right)^M-1.
\end{equation}

Computing the limit of the expression \eqref{I}, when $M$ goes to $+\infty$, it follows that $i=e^r-1$. Hence, it can be establish the equivalences between two discount factors:
\[
e^{-rt}=\dfrac{1}{(1+i)^t}.
\]

\section*{Appendix B: Lambert $W$ Function}\label{B}

Let $z\in\C$ be any complex number and define $W(z)\subset\C$ as a set of complex numbers $w$, such that satisfies the equation
\[
z=we^{w}.
\]

W(z) is called {\it Lambert $W$ Function}. It is defined to be the multivalued inverse of the function $w\mapsto we^w$, for $w\in\C$.

If $z$ is real, then for $-1/e\geq z<0$ there are two real numbers in the set $W(z)$ (see Figure \ref{W}). We denote the branch satisfying $-1<W(z)$ by $W_0(z)$, or just $W(z)$ when there is no possibility for confusion, and the branch satisfying $W(z)\geq -1$ by $W_{-1}(z)$.

\begin{figure}[h]
\psfrag{z}{$z$}
\psfrag{W}{$W(z)$}
\psfrag{E}{$-1/e$}
\includegraphics[width=6cm]{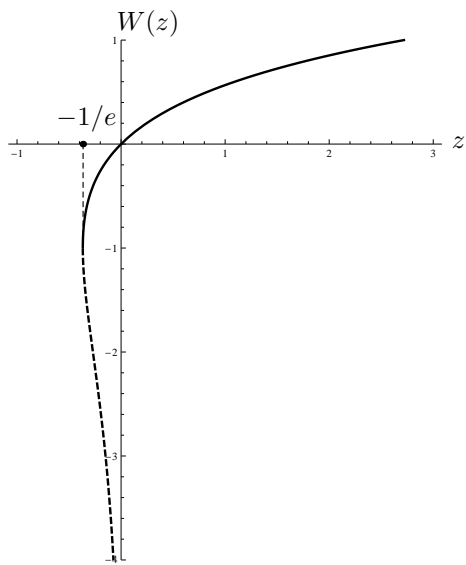}
\vskip 0cm \centerline{} \caption{\small \label{W} The two real branches for $W(z)$. -----, $W_0(z)$; - - -, $W_{-1}(z)$.}
\end{figure}

$W_0(z)$ is referred to as the {\it principal branch} of the $W$ function, and it is analytic at 0. Moreover, the series expansion for $W_0(z)$ is given by
\begin{equation}\label{Ws}
W_0(z)=\sum_{n=1}^{\infty}\dfrac{(-n)^{n-1}}{n!}z^n,
\end{equation}
which has the radius of convergence equal to $1/e$, i.e., the series \eqref{Ws} converges for all $z\in(-1/e,1/e)$.

If we consider the restricted function $W_0:(-1/e,\infty)\longrightarrow(-1,\infty)$, then $W_0$ becomes invertible with $W_0^{-1}:(-1,\infty)\longrightarrow(-1/e,\infty)$ being an increasing function. Indeed, $W_0^{-1}(w)=we^w$, for $w\in(-1,\infty)$.

For more detail on the Lambert $W$ function, see, for instance, \cite{CGHJK}.
 
\section*{Acknowledgements}
The first author is patially suported by an ANEEL grant PD-0064-1020/2010.
The second author is suported by a FAPESP grant 2012/10231-7.

\end{document}